\documentclass[a4paper,UKenglish,cleveref, autoref, thm-restate,authorcolumns]{lipics-v2019}
\usepackage{dsfont}
\usepackage[ruled, linesnumbered]{algorithm2e}
\usepackage{graphicx}
\usepackage{mathtools}
\usepackage{amsfonts}
\usepackage{comment}
\usepackage{color}
\usepackage{pgfplots}
\usepackage{cite}
\pgfplotsset{width=5.6cm,compat=1.9,
    tick label style = {font = {\fontsize{8 pt}{12 pt}\selectfont}},
    label style = {font = {\fontsize{9 pt}{12 pt}\selectfont}},
    legend style = {font = {\fontsize{7 pt}{12 pt}\selectfont}},  
  }
\title{Shortest Path Centrality and the APSP problem via VC-dimension and Rademacher Averages} 

\titlerunning{SP Centrality and the APSP problem via VC-dimension and Rademacher Averages} 

\author{$\,$}{$\,$}{}{}{} 

\author{Alane M. de Lima}{Department of Computer Science, Federal University of Paran\'{a}, Brazil \and \url{http://www.inf.ufpr.br/amlima} }{amlima@inf.ufpr.br}{https://orcid.org/0000-0003-4575-2401}{Supported by CAPES and CNPq.}%TODO mandatory, please use full name; only 1 author per \author macro; first two parameters are mandatory, other parameters can be empty. Please provide at least the name of the affiliation and the country. The full address is optional

\author{Murilo V. G. da Silva}{Department of Computer Science, Federal University of Paran\'{a}, Brazil \and \url{http://www.inf.ufpr.br/murilo}}{murilo@inf.ufpr.br}{http://orcid.org/0000-0002-3392-714X}{}

\author{Andr\'{e} L. Vignatti}{Department of Computer Science, Federal University of Paran\'{a}, Brazil \and \url{http://www.inf.ufpr.br/vignatti}}{vignatti@inf.ufpr.br}{http://orcid.org/0000-0001-8268-5215}{}

\authorrunning{ }
\authorrunning{A.\,M. de Lima et al.} %TODO mandatory. First: Use abbreviated first/middle names. Second (only in severe cases): Use first author plus 'et al.'
\Copyright{ }
\Copyright{Alane M. de Lima, Murilo V. G. da Silva and Andr\'{e} L. Vignatti} %TODO mandatory, please use full first names. LIPIcs license is "CC-BY";  http://creativecommons.org/licenses/by/3.0/

\ccsdesc[100]{Theory of computation}
\ccsdesc[100]{Graph algorithms analysis} %TODO mandatory: Please choose ACM 2012 classifications from https://dl.acm.org/ccs/ccs_flat.cfm 

\keywords{All-pairs Shortest Paths; Shortest Path Centrality; Sample Complexity} %TODO mandatory; please add comma-separated list of keywords

\category{} %optional, e.g. invited paper

\relatedversion{} %optional, e.g. full version hosted on arXiv, HAL, or other respository/website
\supplement{}%optional, e.g. related research data, source code, ... hosted on a repository like zenodo, figshare, GitHub, ...

\nolinenumbers %uncomment to disable line numbering

\hideLIPIcs  %uncomment to remove references to LIPIcs series (logo, DOI, ...), e.g. when preparing a pre-final version to be uploaded to arXiv or another public repository

\begin{document}

\maketitle

\SetKwRepeat{Do}{do}{while}

\begin{abstract}
In this paper we are interested in a version of the All-pairs Shortest Paths problem (APSP) that fits neither in the exact nor in the approximate case. We define a measure of centrality of a shortest path, related to the ``importance'' of such shortest path in the graph, and propose an algorithm based on the idea of progressive sampling that, for {\it any fixed constants} $0 < \epsilon$, $ \delta < 1$, given an undirected graph $G$ with non-negative edge weights, outputs with probability $1 - \delta$ a data structure of size $n \cdot \textrm{Diam}_V(G)$, where $\textrm{Diam}_V(G)$ is the vertex diameter of $G$, in expected time $\mathcal{O}(\lg n \max(m + n \log n, n \cdot \textrm{Diam}_V(G)))$ containing the (exact) distance and the shortest path between every pair of vertices $(u,v)$ that has centrality at least $\epsilon$. The progressive sampling technique is sensitive to the probability distribution of the input (if we assume that $G$ is chosen from a prescribed random distribution), but even in the case where we take no assumption about such distribution, we show an upper bound for the sample size using VC-dimension theory that is tighter than the bound given by standard Hoeffding and union bounds, since VC-dimension theory captures the combinatorial structure of the input graph.
\end{abstract}

%----- 

\section{Introduction}

The All-pairs Shortest Path (APSP) is the problem of computing a path with the minimum length between every pair of vertices in a weighted graph. 
% Comentei a frase abaixo por que para um congresso de grafos isso não precisa e além disso, sempre que fazemos um claim como esse abaixo precisaríamos de referências para sustentar. Então acho que atrapalharia mais do que ajudaria.
%It is a fundamental problem which is the core of many applications.
%and graph parameters, from web mapping \cite{teslya} to social networks analysis \cite{RiondatoUpfal, lima2019}. 
%The APSP problem is extensively studied, both on the exact \cite{pettie2002, pettie2004, williams2014} as well as the approximated case {\color{blue} citar aproximados}. 
The APSP problem is very well studied and there has been recent results for a variety of assumptions for the input graph (directed/undirected, integer/real edge weights, etc) \cite{williams2014, chan2012, eirinakis2017, brodnik2017}. In this paper we assume that the input is an undirected graph $G$ with $n$ vertices and $m$ edges with non-negative weights.
%Comentei por que acho que não precisa...
%{\color{red}(citei algoritmos para o caso exato, não-dinâmico e não distribuído, DEVO CITAR PARA ESTES OUTROS CASOS?)}. 

In our scenario, the fastest known exact algorithms are the algorithm proposed by Williams (2014) \cite{williams2014}, which runs in $\mathcal{O}\left(\frac{n^3}{2^{c \sqrt{\log n}}}\right)$ time, for some constant $c > 0$, and by Pettie and Ramachandram (2002) \cite{pettie2002} for the case of sparse graphs, which runs in $\mathcal{O}(mn \log \alpha(m,n))$ time, where $\alpha(m,n)$ is the Tarjan's inverse-Ackermann function. If no assumption is taken about the sparsity of the graph, it is an open question whether the APSP problem can be solved in $\mathcal{O}(n^{3-c})$, for any $c > 0$ even when the edge weights are natural numbers. The hypothesis that there is no such algorithm for such task is used in hardness arguments in some works\cite{abboudwilliams,abboudwilliams2}. 

The three fastest approximation algorithms for the problem depend on the approximation guarantees as well as the sparsity of the input graph. Elkin et al. (2019) \cite{elkin2019} proposed an approximation that runs in $\mathcal{O}(n^2)$ time and has multiplicative factor $1+\epsilon$ and an additive term $\beta(G)$, where $\beta(G)$ depends on the edge weights. Baswana and Kavitha (2010) \cite{baswana2010} proposed two algorithms, one that runs in $\tilde{\mathcal{O}}(m\sqrt{n} + n^2)$ time and other that runs in $\tilde{\mathcal{O}}(m^{2/3}n + n^2)$ time, depending on the required approximation factor. There has also been recent development on approximation algorithms when assumptions for the input graph are not the same as ours \cite{yuster2012, bringman2019}.

In this paper, we deal with a version of the problem that fits neither in the exact nor in the approximate case. For every pair of vertices $u,v$ our algorithm either outputs a shortest path between $u$ and $v$ of exact size or it does not compute any shortest path between $u$ and $v$, depending on a certain centrality measure of the shortest path in question. The precise definition of centrality is given in Section 2.1.
%, but the intuition is that, for a pair of vertices $(u,v)$, the centrality $c(u,v)$ relates to the amount of shor test paths $P$ in $G$ such that $P$ is a subpath of a shortest path between $u$ and $v$. 
Let $\textrm{Diam}_V(G)$ be the vertex-diameter of the input graph (i.e., the maximum number of vertices in a shortest path in $G$). In this paper we present a $\mathcal{O}(\lg n \max(m + n \log n, n \cdot \textrm{Diam}_V(G)))$ expected running time randomized algorithm for computing the shortest path between every pair of vertices that has shortest path centrality higher than certain fixed constant. This is particularly interesting for sparse graphs with logarithmic vertex diameter (this is the case for many real word graphs such as power law graphs), where the algorithm runs in $\mathcal{O}(n \log^2 n)$ expected time. The central idea is sampling shortest paths trees (such as trees computed by Dijkstra's Algorithm) so that every shortest path in $G$ with the desired centrality is covered with high probability. 

The main contribution of this paper is an analysis relying on the use of Rademacher Averages in a progressive sampling approach to build an algorithm that iteratively increases the sample size until the desired accuracy is achieved. The number of steps in the progressive sampling technique is sensitive to the probability distribution of the input graph (if we assume that the input is sampled according to a certain distribution). However, even if we make no assumption on the input graph, we use the Vapnik-Chervonenkis (VC) theory to give an upper bound for the sample size that is tighter than bounds given by standard Hoeffding and union-bound techniques. This upper bound is tighter since VC dimension theory captures the combinatorial structure of the input graph and this bound for such graph can be computed efficiently. More precisely, we show that sampling $\lceil \frac{c}{\epsilon} (\lfloor \lg \textrm{Diam}_V(G) + \lg n + 2 \rfloor) \ln \frac{1}{\epsilon} + \ln \frac{1}{\delta} \rceil$ vertices and using them as the root of a shortest path trees is enough to compute, with probability $1-\delta$, all shortest paths with centrality at least $\epsilon$, where $c$ is a constant around $\frac{1}{2}$. %\cite{loffler2009shape}.
%More precisely we propose two algorithms. First we show that in $\mathcal{O}(n^2 \log n)$ time a randomized algorithm can estimate a measure $c(u,v)$, for every pair of vertices $(u,v)$, called the \emph{centrality of a shortest path} between $u$ and $v$. The centrality of a shortest path is defined as follows. For a given pair of vertices $(u,v)$ and a given set $\mathcal{C} = C_1 \cup ... \cup C_{|V(G)|}$ of shortest paths of $G$, the value of $c(u,v)$ corresponds to the fraction of sets $C_i$ in $\mathcal{C}$ such that each $C_i$ contains at least one shortest path $P$ between $u$ and $v$.

Some of the techniques used in this paper were developed by Riondato and Kornaropoulos (2016) and Riondato and Upfal (2018) \cite{Riondato2016, RiondatoUpfal}. In their work, the authors use VC-dimension theory, the $\epsilon$-sample theorem and Rademacher averages for the estimation of betweenness centrality in a graph. More recently Lima et al. (2019) \cite{lima2019} used some of these tools for the estimation of the percolation centrality using pseudo-dimension theory. The algorithms in \cite{Riondato2016,RiondatoUpfal,lima2019} use a range space where intervals are vertices of $G$ and the points in sample space are the shortest paths of $G$. A main difference in the approach in the present paper compared to the previous works is that here we show that a ``flipped'' version of the range space (now intervals are shortest paths and points are vertices) can be used for solving the version of the APSP that we have at hand. 

In fact, we show two different algorithms in this paper. The first algorithm outputs with probability $1-\delta$ an estimation for the centrality $c(u,v)$ within $\epsilon$ of the optimal value, for any fixed constants $0 < \epsilon, \delta < 1$. The second algorithm outputs with probability $1-\delta$ a shortest path between $u$ and $v$ if $c(u,v)$ is at least $\epsilon$. Both algorithms run in $\mathcal{O}(\lg n \max(m + n \log n, n \cdot \textrm{Diam}_V(G)))$ expected running time, even though the constant for the second algorithm is smaller.
These constants represent a worst case analysis for the progressive sampling technique, but even in this scenario they are both smaller than the constants obtained by similar sampling algorithms using Hoeffding and union-bound techniques.

\section{Preliminaries}
\label{sec:preliminaries}

The definitions, notation and results which are the theoretical foundation of our work are presented below.

\subsection{Shortest Paths in Graphs} \label{subsec:graphs}

Let $G = (V,E)$ be an undirected graph and a function $w : E \rightarrow \mathbb{R}^+$, where $w(e)$ is the non-negative \emph{weight} of the edge $e$. W.l.o.g. we assume that $G$ is connected, since all results in this paper can be applied to the connected components of a graph, when a graph is disconnected.  A \emph{path} is a sequence of vertices $p = (v_1, v_2, \ldots, v_k)$ such that, for $1 \leq i < k$, $v_i \neq v_{i+1}$ and there is $(v_i,v_{i+1}) \in E$. Let $E_p$ be the set of edges of a path $p$. The length of $p$, denoted by $l(p)$, corresponds to the sum $\sum_{e \in E_P} w(e)$. For a pair $(u,v) \in V^2$, let $P_{uv}$ be the set of all paths from $u$ to $v$. A \emph{shortest path} is a path $p_{uv} \in P_{uv}$ where  $l(p_{uv}) = \min\{l(p_{u'v'}) : p_{u'v'} \in P_{uv}\}$. The length of a shortest path is called \emph{distance}.

A \emph{shortest paths tree (SPT)} of a vertex $u$ is a spanning tree $T_u$ of $G$ such that the path from $u$ to every other vertex of $T_u$ is a shortest path in $G$. Note that there might be many SPTs for a given vertex. In this paper we are interested in fixing one canonical SPT for every vertex of $G$. More precisely, we fix an (arbitrary) ordering of the vertex set $V$ and let the canonical SPT for a vertex $u$ be the SPT output by Dijkstra's algorithm and denote such tree $T_u$. We also call $T_u$ the \emph{Dijkstra tree} of $u$. Let $p_{uv}$ be a shortest path from the root $u$ to $v$ in the tree $T_u$. Then every subpath of $p_{uv}$ is also a shortest path in $G$. We denote such set of subpaths (including $p_{uv}$) as $S(p_{uv})$. Since $G$ is undirected, the same applies for paths in reverse order, i.e., every subpath of $p_{vu}$ in $T_u$ is also a shortest path. Let $S(p_{vu})$ be such set of shortest paths.

%mvgs: mexi abaixo também para que o conjunto de caminhos mínimos da árvore tenha a ida e a volta. Antes era $S(T_u) = \bigcup_{v \in V \setminus \{u\}} S(p'_{uv})$. Agora é $S(T_u) = \bigcup_{v \in V \setminus \{u\}} (S(p'_{uv}) \cup S(p'_{vu}))$.

Note that there are exactly $n$ Dijkstra trees for $G$ since Dijkstra's algorithm is deterministic and we have a fixed ordering for $V$. The set of $n$ Dijkstra trees of $G$ is denoted by $\mathcal{T}$.  Let $S(T_u) = \bigcup_{v \in V \setminus \{u\}} (S(p_{uv}) \cup S(p_{vu}))$.  The canonical set of shortest paths of $G$ (w.r.t. the ordering) is $S(G) = \bigcup_{u \in V} S(T_{u})$.   For the sake of convenience in Definition \ref{def:lengthshortpath} we present the length of a shortest path (distance) between a pair $(u,v) \in V^2$ in terms of Dijkstra trees.  

\begin{definition}[Distance]\label{def:lengthshortpath}
    Given a graph $G = (V,E)$, a function $w : E \rightarrow \mathbb{R}^+$, a pair $(u,v) \in V^2$, a vertex $x$ and a Dijkstra tree $T_x$, the \emph{distance} from $u$ to $v$ is defined as
    
    $$d(p_{uv}) = \sum_{e \in E_{T_x}} w(e)$$
    where $p_{uv} \in S(T_x)$ and $E_{T_x}$ is the set of edges of $p_{uv}$ in $T_x$.
\end{definition}

We define below the \emph{shortest path centrality} as the proportion of Dijkstra trees that contains some shortest path from $u$ to $v$ among all $n$ Dijkstra trees.

\begin{definition}[Shortest Path Centrality]\label{def:shortestpathcentrality} Given an undirected weighted graph $G = (V,E)$ with $n = |V|$, a pair $(u,v) \in V^2$ and the Dijkstra tree $T_x$ for each $x \in V$, let $p_{uv} = (u,\ldots,v)$ be a shortest path from $u$ to $v$ such that $p_{uv} \in S(G)$. The \emph{shortest path centrality} of a pair $(u,v) \in V^2$ is defined as
$c(u,v) = \frac{t_{uv}}{n}$
where $t_{uv} = \sum_{T_x \in \mathcal{T}} \mathds{1}_{\tau_{uv}}(T_x)$ and $\tau_{uv} = \{T_x \in \mathcal{T} : p_{uv} \in S(T_x)\}$.
\end{definition}

The function $\mathds{1}_{\tau_{uv}}(T_x)$ returns 1 if there is some shortest path from $u$ to $v$ in $T_x$ (and 0 otherwise). Clearly the centrality measure depends on the fixed ordering, however, since we are dealing random sampling note that such ordering is irrelevant. 

% Comentei a frase abaixo por falta de espaço. Pode voltar no journal
%Intuitively speaking, a pair $(u,v)$ has high shortest path centrality if the canonical shortest path $p_{uv} \in S(T_u)$ (and $S(T_v)$) is a subpath of a large number of canonical shortest paths in $S(G)$.

%---------------------------------------------------------------------------------
\subsection{Sample Complexity and VC-dimension}\label{subsec:samplecomp}

% esse parágrafo não tinha no WG
In sampling algorithms, typically the aim is the estimation of a certain quantity according to given parameters of quality and confidence using a random sample of size as small as possible. A central concept in sample complexity theory is the Vapnik-Chervonenkis Theory (VC-Dimension), in particular, the idea of finding an upper bound for the VC-dimension of a class of binary functions related to the sampling problem at hand. In our context, for instance, we may consider a binary function that takes a Dijkstra tree and outputs 1 if such tree contains a shortest path for a given set. Generally speaking, from the upper bound for the given class of binary functions we can derive an upper bound to the sample size for the sampling algorithm. 

We present in this section the main definitions and results from sample complexity theory used in this paper. An in-depth exposition of the Vapnik-Chervonenkis Theory (VC-Dimension), the $\epsilon$-sample and the $\epsilon$-net theorems can be found in the books of Shalev-Schwartz and Ben-David (2014) \cite{ShalevShwartz2014}, Mitzenmacher and Upfal (2017) \cite{Mitzenmacher2017}, Anthony and Bartlett (2009) \cite{Anthony2009}, and Mohri et al. (2012) \cite{Mohri2012}. 

\begin{definition}[Range Space]
    A \emph{range space} is a pair $\mathcal{R} = (U, \mathcal{I})$, where $U$ is a domain (finite or infinite) and $\mathcal{I}$ is a collection of subsets of $U$, called \emph{ranges}.
\end{definition}

 For a given $S \subseteq U$, the \emph{projection} of $\mathcal{I}$ on $S$ is the set $\mathcal{I}_S = \{S \cap I : I \in \mathcal{I}\}$. If $|\mathcal{I}_S| = 2^{|S|}$ then we say $S$ is \emph{shattered} by $\mathcal{I}$. The VC-dimension of a range space is the size of the largest subset $S$ that can be shattered by $\mathcal{I}$, i.e., 

\begin{definition}[VC-dimension]\label{def:vcdim}
    The VC-dimension of a range space $\mathcal{R} = (U,\mathcal{I})$, denoted by $\textrm{VCDim}(\mathcal{R})$, is 
    $\textrm{VCDim}(\mathcal{R}) = \max\{k : \exists S \subseteq U \text{ such that } |S| = k \text{ and } |\mathcal{I}_S| = 2^k\}$.
\end{definition}

The following combinatorial object, called $\epsilon$-net,
%models applications of VC-dimension in 
is useful when one wants to find a sample $S \subseteq U$ that intersects every range in $\mathcal{I}$ of a sufficient size. 
%More specifically, 

\begin{definition}[$\epsilon$-net]\label{enet}
    Given $0 < \epsilon < 1$, a set $S$ is called $\epsilon$-net w.r.t. a range space $\mathcal{R} = (U, \mathcal{I})$ and a probability distribution $\pi$ on $U$ if 
    \[ \forall I \in \mathcal{I}, \quad \Pr_\pi(I) \geq \epsilon \Rightarrow |I \cap S| \geq 1. \]
\end{definition}

The definition of $\epsilon$-sample is a stronger notion as it not only intersects ranges of a sufficient size but it also guarantees the right relative frequency of each range in $\mathcal{I}$ within the sample $S$. 

\begin{definition}[$\epsilon$-sample]\label{esample}
    Given $0 < \epsilon < 1$, a set $S$ is called $\epsilon$-sample w.r.t. a range space $\mathcal{R} = (U, \mathcal{I})$ and a probability distribution $\pi$ on $U$ if 
    \[ \forall I \in \mathcal{I}, \quad \left| \Pr_\pi(I) - \frac{|S \cap I|}{|S|} \right| \leq \epsilon. \]
\end{definition}

A more general definition of $\epsilon$-sample (called $\epsilon$-representative) is given below for the context where for a given a domain $U$ and a set of values of interest $\mathcal{H}$, there is a family of functions $\mathcal{F}$ from $U$ to $\mathbb{R}^*$ such that there is one $f_h \in \mathcal{F}$ for each $h \in \mathcal{H}$. Let $S$ be a collection of $r$ elements from $U$ sampled with respect to a probability distribution $\pi$. 

\begin{definition}\label{def:averages}
    For each $f_h \in \mathcal{F}$, such that $h \in \mathcal{H}$, we define the expectation of $f_h$ and its empirical average as $L_U$ and $L_S$, respectively, i.e., 
\[L_U(f_h) = \mathbb{E}_{u \in U} [f_h(u)] \quad \text{and} \quad L_S(f_h) = \frac{1}{r} \sum_{s \in S} f_h(s).\]
\end{definition}

\begin{definition}\label{def:erepresentative}
    Given $0 < \epsilon,\delta < 1$, a set $S$ is called \emph{$\epsilon$-representative} w.r.t. some domain $U$, a set $\mathcal{H}$, a family of functions $\mathcal{F}$ and a probability distribution $\pi$ if $|L_S(f_h) - L_U(f_h)| \leq \epsilon$, $\forall f_h \in \mathcal{F}.$
\end{definition}

%ALANE: def. do Shalev-Schwartz / ver pag. 56 do livro Understanding ML

The expectation of the empirical average $L_S(f_h)$ corresponds to $L_U(f_h)$, by linearity of expectation, therefore $|L_S(f_h) - L_U(f_h)| = |L_S(f_h) - \mathbb{E}_{f_h \in \mathcal{F}}[L_S(f_h)]|$. By the \emph{law of large numbers}, $L_S(f_h)$ converges to its true expectation as $r$ goes to infinity, since $L_S(f_h)$ is the empirical average of $r$ random variables sampled independently and identically w.r.t. $\pi$. Since this law provides no information about the value $|L_S(f_h) - L_U(f_h)|$ for any sample size, we use results from VC-dimension theory, which provide bounds on the size of the sample that guarantees that the maximum deviation of $|L_S(f_h) - L_U(f_h)|$ is within $\epsilon$ with probability at least $1-\delta$, for given $0 < \epsilon,\delta < 1$.

%aml: Ver a definição de pseudo-dimension. No STACS o intervalo era [0,1], mas até onde lembro funciona para R^*.
%\paragraph*{Pseudo-dimension} 

   %Let $\mathcal{F}$ be a family of functions from some domain $U$ to the $\mathbb{R}^*$. Consider $D = U \times \mathbb{R}^*$. For each $f \in \mathcal{F}$, there is a subset $R_f \subseteq D$ defined as $R_f = \{(x,t) : x \in U \text{ and } t \leq f(x)\}$. The next definition {\color{red} def ou teo?} is from \cite{Anthony2009}, Section 11.2

%\begin{definition} [Pseudo-dimension] \label{def:pseudo}
 %   Let $\mathcal{R} = (U,\mathcal{F})$ and $\mathcal{R}' = (D,\mathcal{F}^+)$ be range spaces, where $\mathcal{F}^+ = \{R_f : f \in \mathcal{F}\}$. The \emph{pseudo-dimension} of $\mathcal{R}$, denoted by $PD(\mathcal{R})$, corresponds to the VC-dimension of $\mathcal{R}'$, i.e., $PD(\mathcal{R}) = VCDim(\mathcal{R}')$. 
%\end{definition}

An upper bound to the VC-dimension of a range space allows to build an $\epsilon$-net and an $\epsilon$-representative sample, as stated in Theorem \ref{teo:esample}.

\begin{theorem}[see \cite{har2011}, Theorem 2.12] \label{teo:esample}
    Given $0 < \epsilon,\delta < 1$, let $\mathcal{R} = (U, \mathcal{I})$ be a range space %($D = U \times \mathbb{R}^*$) 
    with $\textrm{VCDim}(\mathcal{R}) \leq k$, a probability distribution $\pi$ on the domain $U$ and let $c$ be a universal positive constant. 
    \begin{enumerate}
        \item A collection of elements $S \subseteq U$ sampled w.r.t. $\pi$ with $|S| = \frac{c}{\epsilon^2} \left( k + \ln \frac{1}{\delta} \right)$ is $\epsilon$-representative with probability at least $1-\delta$.
        
        \item A collection of elements $S \subseteq U$ sampled w.r.t. $\pi$ with $|S| = \frac{c}{\epsilon} \left( k \ln \frac{1}{\epsilon} + \ln \frac{1}{\delta} \right)$ is an $\epsilon$-net with probability at least $1-\delta$.
    \end{enumerate}
\end{theorem}

As pointed by Löffler and Phillips (2009) \cite{loffler2009shape}, $c$ is around $\frac{1}{2}$, but in this paper we leave $c$ as an unspecified constant. 

 %Lemmas \ref{lemma:atmostone} and \ref{lemma:nozero}, stated an proved by Riondato and Upfal (2018), present constraints on the sets that can be shattered by a range set $\mathcal{F}^+$. 

%\begin{lemma}[see \cite{RiondatoUpfal}, Section 3.3] \label{lemma:atmostone}
 %   Let $B \subseteq D$ be a set that is shattered by $\mathcal{F}^+$. Then, $B$ can contain at most one $(d,y) \in D$ for each $d \in U$ and for a $y \in [0,1]$.
%\end{lemma}

%\begin{lemma}[see \cite{RiondatoUpfal}, Section 3.3] \label{lemma:nozero}
 %   Let $B \subseteq D$ be a set that is shattered by $\mathcal{F}^+$. Then, $B$ does not contain any element in the form $(d,0) \in D$, for each $d \in U$.
%\end{lemma}

%---------------------------------
\subsection{Progressive Sampling and Rademacher Averages} \label{subsec:rademacher}

Finding a bound to the sample size that is tight may be a complicated task depending on the problem. Hence, making use of progressive sampling, in which the process starts with a small sample size which progressively increases until the accuracy improves, becomes an alternative to this issue \cite{provost1999}. The combination of an appropriate scheduling for the sample increase with an efficient-to-evaluate stopping condition (i.e., knowing when the sample is large enough) leads to a greater improvement in time for the estimation of the value of interest \cite{RiondatoUpfal}.
%In our algorithm, presented in Section \ref{sec:algapprox}, we sample Dijkstra trees in $G$, and the bound given by Theorem \ref{theo:vcdimproblem} limits the sample size in terms of the VC-dimension of a range space associated to $G$. Additionally, we improve the algorithm making  In this section we present the theoretical framework used in the progressive sampling algorithm.
A key ideia is that the stopping condition takes into consideration the input distribution, which can be extracted by the use of Rademacher Averages (\cite{Mitzenmacher2017}, chapter 14). This theory lies in the core of statistical learning theory, although their applications extend the context of learning algorithms \cite{riondato2015}. %One of the main advantages of using Rademacher averages, besides deriving the bounds from a sample of the domain, is that it applies to any real-valued function and not only to 0--1 classification functions \cite{RiondatoUpfal}.

Consider the computation of the maximum deviation of $L_S(f_h)$ from the true expectation of $f_h$, for all $f_h \in \mathcal{F}$, that is, $\sup_{f_h \in \mathcal{F}} |L_S(f_h) - L_U(f_h)|$. The \emph{empirical Rademacher average} of $\mathcal{F}$ is defined as follows.

\begin{definition}\label{empiricalrademacher}
    Consider a sample $S = \{z_1,\ldots,z_r\}$ and a distribution of $r$ Rade\-macher random variables $\sigma = (\sigma_1,\ldots,\sigma_r)$, i.e., $\Pr(\sigma_i = 1) = \Pr(\sigma_i = -1) = 1/2$ for $1 \leq i \leq r$. The empirical Rademacher average of a family of functions $\mathcal{F}$ w.r.t. to $S$ is defined as
    
    $$\tilde{R}_r(\mathcal{F},S) = \mathbb{E}_\sigma \left[ \sup_{f_h \in \mathcal{F}} \frac{1}{r} \sum_{i=1}^r \sigma_i f_h(z_i) \right].$$
\end{definition}

In this work, we use the bound previously introduced by Riondato and Upfal \cite{RiondatoUpfal} for the connection of the empirical Rademacher average with the value of $\sup_{f_h \in \mathcal{F}} |L_S(f_h) - L_U(f_h)|$, which extended the bound of Oneto et al. \cite{oneto2013} to the supremum of its absolute value to functions with codomain in $[0,1]$.

\begin{theorem}\label{theo:rademacherbound}
    With probability at least $1-\delta$,
    
    $$\sup_{f_h \in \mathcal{F}} |L_S(f_h) - L_U(f_h)| \leq 2\tilde{R}_r(\mathcal{F},S) + \frac{\ln \frac{3}{\delta} + \sqrt{(\ln \frac{3}{\delta} + 4r \tilde{R}_r(\mathcal{F},S)) \ln \frac{3}{\delta}}}{r} + \sqrt{\frac{\frac{3}{\delta}}{2r}}.$$
    
\end{theorem}

The exact computation of $\tilde{R}_r(\mathcal{F},S)$ depends on an extreme value, i.e., the supremum of deviations for all functions in $\mathcal{F}$, which can be expensive over a large (or infinite) set of functions \cite{Mitzenmacher2017}. Even a Monte Carlo simulation to estimating $\tilde{R}_r(\mathcal{F},S)$ is expensive to be extracted in this case; hence, we use the bound given by Theorem \ref{theo:massart}, which is a variant of the Massart's Lemma (see Theorem 14.22, \cite{Mitzenmacher2017}) that is convex, continuous in $\mathbb{R}^+$ and can be efficiently minimized by standard convex optimization methods. 

Consider the vector $v_{f_h} = (f_h(z_1),\ldots,f_h(z_m))$ for a given sample of $m$ elements, denoted by $S = \{z_1,\ldots,z_m\}$, and let $\mathcal{V}_S = \{v_{f_h}, f_h \in \mathcal{F}\}$.

\begin{theorem}{(Riondato and Upfal \cite{RiondatoUpfal})}\label{theo:massart}
    Let $w : \mathbb{R}^+ \rightarrow \mathbb{R}^+$ be the function
    
    $$w(s) = \frac{1}{s} \ln \sum_{v_{f_h} \in \mathcal{V}_S} \exp{\left(\frac{s^2 ||v_{f_h}||_2^2}{2m^2}\right)}.$$
    
    Then $\tilde{R}_r(\mathcal{F},S) \leq \min_{s \in \mathbb{R}^+} w(s).$
\end{theorem}

%----------------------------------
\section{Estimation for the Shortest Path Centrality and the All-pairs Shortest Path Problem}\label{sec:algapprox}

We first define the problem in terms of a range space, and then we give an outline of the two algorithms that we present in this paper. Both algorithms take as input an undirected graph $G = (V,E)$ with $n$ vertices and $m$ edges with non-negative edge weights, a sample schedule $(|S_i|)_{i \geq 1}$ and the quality and confidence parameters $0 < \epsilon, \delta < 1$, assumed to be constants. 

Let $n = |V|$ and $\mathcal{T}$ be the set of $n$ Dijkstra trees of $G$. The set $\mathcal{H}$ from Section \ref{subsec:samplecomp} is defined to be $V^2$ and the universe $U$ is the set $\mathcal{T}$.
For each pair $(u,v) \in V^2$, let $p_{uv}$ be a shortest path from $u$ to $v$. Let $\tau_{uv} = \{T_x \in \mathcal{T} : p_{uv} \in S(T_x)\}$. Let $\mathcal{I} = \{\tau_{uv} : (u,v) \in V^2\}$. Note that $\mathcal{R} = (\mathcal{T}, \mathcal{I})$ is a range space. For $T_x \in \mathcal{T}$, let $f_{uv} : \mathcal{T} \rightarrow \{0,1\}$ be the function $f_{uv}(T_x) = \mathds{1}_{\tau_{uv}}(T_x).$ The indicator function $\mathds{1}_{\tau_{uv}}(T_x)$ returns 1 if there is some shortest path from $u$ to $v$ in $T_x$ (and 0 otherwise).
%, once the canonical shortest paths in $T_x$ are directed. 
We define $\mathcal{F} = \{f_{uv} : (u,v) \in V^2\}$.

Each $T_x \in \mathcal{T}$ is sampled according to the function $\pi(T_x) = \frac{1}{n}$ (which is a valid probability distribution), %as stated by Theorem \ref{theo:probdist}), 
and $\mathbb{E}[f_{uv}(T_x)] = c(u,v)$ for all $(u,v) \in V^2$, as proved in Theorem \ref{theo:expec}.

\begin{theorem}\label{theo:expec}
    For $f_{uv} \in \mathcal{F}$ and for $T_x \in \mathcal{T}$, such that each $T_x$ is sampled according to the probability function $\pi(T_x)$, $\mathbb{E}[f_{uv}(T_x)] = c(u,v).$
\end{theorem}

\begin{proof}
    Given an undirected weighted graph $G = (V,E)$, for all $(u,v) \in V^2$, we have from Definition \ref{def:averages}
    
%mvgs: IMPORTANTE: coloquei uma igualdade a mais: L_U(f_{uv}) = L_{\mathcal{T}}(f_{uv})
%mvgs: IMPORTANTE: além disso eu troquei os demais U por \mathcal{T}. É pra dar na mesma, certo? Ainda assim eu mantive o primeiro L_U para linkar com a seção de VC-dimension.
    
    \begin{align*}
        L_U(f_{uv}) = L_{\mathcal{T}}(f_{uv}) = \mathbb{E}_{T_x \in \mathcal{T}}[f_{uv}(T_x)] &= \sum_{T_x \in \mathcal{T}} \pi(T_x)f_{uv}(T_x) \\
        %&= \sum_{T_x \in \mathcal{T}} \frac{1}{n} \mathds{1}_{\tau_{uv}}(T_x) \text{\color{red} \ [cortar essa linha para ganhar espaço?]} \\
        &= \frac{1}{n} \sum_{T_x \in \mathcal{T}} \mathds{1}_{\tau_{uv}}(T_x) = \frac{t_{uv}}{n} = c(u,v)
    \end{align*}
    
    %The equation in (4) follows from the assumption of the uniqueness of the shortest paths, and equation in (5) follows from the fact that the indicator function $\mathds{1}_{\tau_{uv}}(e,T_x)$ returns 1 only if $x = u$ (and then $T_x = T_u$) and if $e$ is in the shortest path from $u$ to $v$ in $T_u$.
\end{proof}

%mvgs: aqui abaixo tbm mudei de U para \mathcal{T}

Let $S = \{T_i, 1 \leq i \leq r\}$ be a set of $r$ Dijkstra trees sampled independently and identically from $\mathcal{T}$. Next, we define $\tilde{c}(u,v)$, the estimation to be computed by the algorithm, as the empirical average from Definition \ref{def:averages}:

$$\tilde{c}(u,v) = L_S(f_{uv}) = \frac{1}{r} \sum_{T_i \in S} f_{uv}(T_i) = \frac{1}{r} \sum_{T_i \in S} \mathds{1}_{\tau_{uv}}(T_i).$$

For each $(u,v) \in V^2$, the value $\tilde{c}(u,v)$ can be defined as $||\textbf{v}_{uv}||_1/r$, where $$\textbf{v}_{uv} = (f_{uv}(T_1),\ldots,f_{uv}(T_r)).$$ Each function $f_{uv}$, however, is a binary function such that $||\textbf{v}_{uv}|| = t_{uv}$. Hence, we denote $\mathcal{V}$ as the set of such values, i.e., $\mathcal{V} = \{t_{uv}, (u,v) \in V^2\}$. Note that $|\mathcal{V}| \leq (|V^2|-n)/2$, since $G$ is undirected -- and then for a pair $(u,v)$, $t_{uv} = t_{vu}$ -- and there may be different pairs of vertices $(u_k,v_k)$ and $(u_l,v_l)$ with $t_{u_k v_k} = t_{u_l v_l}$.

% WG
%The correctness and running time of our algorithm relies on the sample size given by Theorem \ref{teo:esample}. The sample size is bounded by the result stated in Theorem \ref{theo:vcdimproblem}, which is an upper bound for VCDim($\mathcal{R})$. In remainder of this paper let $Diam_V(G)$ be the vertex-diameter of $G$, i.e., the maximum number of vertices in a shortest path of $G$.

The VC-dimension of the range space $\mathcal{R} = (\mathcal{T},\mathcal{I})$, which is an upper bound to the fixed sample size that guarantees that $|\tilde{c}(u,v) - c(u,v)| \geq \epsilon$ with probability at least $1-\delta$, for each $(u,v) \in V^2$ and for $0 < \epsilon, \delta < 1$, is stated below. In remainder of this paper let $\textrm{Diam}_V(G)$ be the vertex-diameter of $G$, i.e., the maximum number of vertices in a shortest path of $G$.

\begin{theorem}\label{theo:vcdimproblem}
    The VC-dimension of the range space $\mathcal{R} = (\mathcal{T}, \mathcal{I})$ is
    $$\textrm{VCDim}(\mathcal{R}) \leq \lfloor \lg \textrm{Diam}_V(G) + \lg n + 2 \rfloor.$$
\end{theorem}

\begin{proof}
    Let $\textrm{VCDim}(\mathcal{R}) = k$, where $k \in \mathbb{N}$. Then, there is $S \subseteq U$ such that $|S| = k$ and $S$ is shattered by $\mathcal{I}$. By the definition of shattering, each $T_i \in S$ must appear in $2^{k-1}$ different ranges in $\mathcal{I}$. On the other hand, let $\textrm{Diam}_V(T_i)$ be the vertex-diameter of $T_i$. Then $|S(T_i)| = \sum_{v \in V \setminus \{i\}} 2 \cdot h(v) \leq 2 \sum_{v \in V \setminus \{i\}} \textrm{Diam}_V(T_i) = 2n \cdot \textrm{Diam}_V(T_i)$, where $S(T_i)$ contains the shortest paths in $T_i$ as defined in Section \ref{subsec:graphs} and $h(v)$ is the height of $v$ in $T_i$. We have that $|S(T_i)| \leq 2n \cdot \textrm{Diam}_V(G)$,  since in the worst case a Dijkstra tree $T_i$ cannot be deeper than the vertex-diameter of $G$. Hence, $2^{k-1} \leq 2n \cdot \textrm{Diam}_V(G)$, and $k-1 \leq \lg \textrm{Diam}_V(G) + \lg n + 1$.
    %\begin{align*}
      %  2^{k-1} &\leq n^2\\
      %  k-1 & \leq 2 \lg n
    %\end{align*}
    Since $k$ must be integer, $k \leq \lfloor \lg \textrm{Diam}_V(G) + \lg n + 2 \rfloor \leq \lg \textrm{Diam}_V(G) + \lg n + 2$. Finally, $\textrm{VCDim}(\mathcal{R}) = k \leq \lfloor \lg \textrm{Diam}_V(G) + \lg n + 2 \rfloor.$
\end{proof}

Note that for a sample of size $r$, by Hoeffding bound we have $$\Pr(|\tilde{c}(u,v) - c(u,v)| \geq \epsilon) \leq 2exp(-2r\epsilon^2).$$
Applying the union bound for all $(u,v) \in V^2$, the value of $r$ must be $2exp(-2r^2) n^2 \geq \delta$, which leads to $r \geq \frac{1}{2\epsilon^2}(\ln 2 + 2 \ln n + \ln(1/\delta))$. Even though $\textrm{Diam}_V(G)$ might be as large as $n$, we note that our bound given in Theorem \ref{theo:vcdimproblem} is tighter since it depends on the combinatorial structure of $G$. In particular, since a bound on $\textrm{Diam}_V(G)$ can be computed efficiently, in our algorithm we compute a sample size tailored for the input graph in question. For instance, if $\textrm{Diam}_V(G) = \log n$ (which is common in many real-world graphs, in particular power-law graphs), we have that $\textrm{VCDim}(\mathcal{R}) \leq \lfloor \lg (\log n) + \lg n + 2 \rfloor$ which give better constants for the sample size in our sampling algorithm. 

The results above can be improved if we consider a progressive sampling approach instead of running the sampling algorithm directly in a sample of fixed size. We define the progressive sampling schedule of this work as follows: let $S_1$ be the initial sample size and $\delta_1 = \delta/2$. At this point, the only information available about the empirical Rademacher complexity of $S_1$ is that $\tilde{R}_r(\mathcal{F},S_1) \geq 0$. Plugging this with the r.h.s. of the bound in Theorem \ref{theo:rademacherbound}, which has to be at most $\epsilon$, we have

$$\frac{\ln(3/(\delta/2))+\sqrt{\ln(3/(\delta/2)) \ln(3/(\delta/2))}}{|S_1|} + \sqrt{\frac{\ln(3/(\delta/2))}{2|S_1|}} \leq \epsilon$$

$$\frac{2 \ln(6/\delta)}{|S_1|} + \sqrt{\frac{\ln(6/\delta)}{2|S_1|}} \leq \epsilon$$

$$\frac{4 \ln^2(6/\delta)}{|S_1|^2} + \frac{\ln(6/\delta)}{2|S_1|} \leq \epsilon^2$$

which leads to the quadratic inequality
$$2|S_1|^2\epsilon^2 - |S_1|\ln(6/\delta) - 8\ln^2(6/\delta) \geq 0.$$ 

%Applying Bhaskara theorem {\color{red}(acho que não se usa essa forma de falar)}, then 
with solution
$$|S_1| \geq \frac{\ln(6/\delta)(1+\sqrt{1+8^2 \epsilon^2})}{4\epsilon^2}. \quad (1)$$

There is no fixed strategy for scheduling. In our algorithm we follow the results of Provost et al. \cite{provost1999} as well as Riondato and Upfal \cite{RiondatoUpfal} where a geometric sampling schedule is proposed as a strategy (the authors conjecture that such strategy is optimal, but we do not need such assumption), i.e., the one that $S_i = c^i S_1$, for each $i \geq 1$ and for a constant $c > 1$.

Given $0 < \epsilon, \delta < 1$, let $(|S_i|)_{i \geq 1}$ be a geometric sample schedule with starting sample size defined in Equation (1). We present the outline of the algorithms for estimating the shortest path centrality and for the computation of shortest paths with such centrality at least $\epsilon$. Both algorithms return the correct output with probability $1 - \delta$.  For instance, the first algorithm outputs a table $\tilde{c}$ with the centrality estimation, while the second outputs a table $d$ with the distances between the pair of vertices and it stops sampling if the sample size reach the bound in Theorem \ref{teo:esample} (ii).  Since both algorithms are similar, we outline them in parallel.  %(they do not depend on the size of $G$). 

Consider the table for the estimation of canonical shortest paths tree $\tilde{t}$ and the set $\mathcal{V}$ that contains the values in $\tilde{t}$ without repetition. 
%At the beginning all entries of $\tilde{c}$ and $\tilde{t}$ are set to zero and all entries of $d$ are set to $\infty$. 
The following steps are repeated for each $i \geq 1$. For the sake of clarity, $S_0 = \emptyset$.

\begin{description}
    \item[step 1.] Create a sample of $k = |S_i| - |S_{i-1}|$ elements of $V$ chosen uniformly and independently (with replacement) at random;
    \item[step 2.] Sample a vertex $x$, compute a canonical shortest path tree $T_x$ and an array of distances $dist$ of size $n-1$ from $x$ to each $y \in V$, $x \neq y$. For every shortest $p_{uv}$ in $S(T_x)$, increase the canonical shortest path tree value $\tilde{t}(u,v)$ by $1$. Repeat this step $k$ times;
    \item[step 3.] Compute the bound to $\tilde{R}_r(\mathcal{F},S_i)$ by minimizing the function defined in Theorem \ref{theo:massart}. If it satisfies the stopping condition defined in Theorem \ref{theo:rademacherbound} or if the sample size corresponds to the bound in Theorem \ref{teo:esample} (ii), then return the set $\{\tilde{d}(u,v), (u,v) \in V^2 \text{ and } \tilde{d}(u,v) > 0\}$ (and also $\{\tilde{c}(u,v) = \tilde{t}_{uv}/|S_i|, (u,v) \in V^2 \text{ and } \tilde{t}_{uv} > 0\}$ if the estimation $\tilde{c}$ is being computed). Otherwise, increase the size of $S_i$ until it has size $|S_{i+1}|$, increase $i$ and return to step 1.
\end{description}

Step 1 is trivial and step 2 can be performed by running Dijkstra's Algorithm in time $\mathcal{O}(m + n \log n)$ in the input graph $G$ for each sampled vertex $x$. The update of tables $\tilde{t}$ and $d$ is performed by a modification of a DFS algorithm running on $T_x$ with starting vertex $x$. This modification of a DFS can be implemented in the following way. When recursively traversing $T_x$, keep a list $L$ of the predecessors of each visited vertex $v$. So when a vertex $v$ in $T_x$ is visited by the DFS, we have that every vertex $u$ in the list $L$ is a starting vertex of a shortest path from $u$ to $v$ in $S(T_x)$. Since $T_x$ is undirected, we can update both $d(u,v)$ and $d(v,u)$ (and also $\tilde{c}(u,v)$ and $\tilde{c}(v,u)$). Note that since $S(T_x)$ might have $\mathcal{O}(n \cdot \textrm{Diam}_V(G))$ shortest paths, this algorithm perform $\mathcal{O}(n \cdot \textrm{Diam}_V(G))$ updates in the output table. The dimension of the tables $d$, $\tilde{c}$ and $\tilde{t}$ are $\min(n \cdot \textrm{Diam}_V(G)\lg n, n^2)$, since in the worst case $\textrm{Diam}_V(G) = n$ and a value is not reinserted in the table during the traversing of the modified DFS.

%aml: ESSA PARTE NÃO VAI FICAR ASSIM MAIS POR CONTA DA EPSILON-NET
%A key observation is that at the moment that the modified DFS updates $d(u,v)$, the algorithm has the shortest path $p_{uv}$ at hand. So we can both compute a matrix $d$ of shortest distances as well as, at the end of the execution, recover a shortest path for every pair of vertices $(u,v)$ whenever the entry $\tilde{t}(u,v)$ has been updated. 

Next we give the DFS modification and the main algorithm (Algorithms \ref{alg:distances1} and \ref{alg:distances2}, respectively) in detail. For the storage of each value in $\mathcal{V}$ in a sparse way and without repetition, Algorithm \ref{alg:distances1} keeps an array $count$ of size $n$, such that each value in $count[p]$ contains the amount of pairs of vertices having $p$ canonical shortest path trees, for $1 \leq p \leq n$. 
For the sake of clarity, we present in Algorithm \ref{alg:distances2} the computation of table $d$ with the option to compute table $\tilde{c}$ if the indicator variable for the shortest path centrality $spc$ is equal to one.%, where $\tilde{c} = \tilde{t}/|S_r|$, where $r$ is the last iteration of the sampling schedule. 
%Depending on the application one might want to adapt the algorithm for computing only matrix $d$.

\begin{algorithm}[!htbp]
		\SetAlgoNoLine
		\SetAlgoNoEnd
		\DontPrintSemicolon	
		\KwData{List $L$, array of distances $dist$ generated by Dijkstra algorithm, vertex $i$, Dijkstra tree $T_x$, exact distances table $d$, number of canonical shortest path trees table $\tilde{t}$, set $\mathcal{V}$, counter array $count$.}
		
		\KwResult{Updates the distances $d$ and the values in $t$ using paths from $S(T_x)$.}
	    \For{$j \in L$}{
        
		        $d[j][i] \leftarrow d[i][j] \leftarrow dist[i] - dist[j]$\;
		        
		        \If{\text{first time insertion of j and i in} $\tilde{t}$}{
		            $\tilde{t}[j][i] \leftarrow \tilde{t}[i][j] \leftarrow 0$\;
		        }
		        $\tilde{t}[j][i] \leftarrow \tilde{t}[i][j] \leftarrow \tilde{t}[i][j] + 1$\;
		        
		        $count[\tilde{t}[j][i]] \leftarrow count[\tilde{t}[j][i]] + 1$\;
		        
		        \If{($\tilde{t}[j][i]$ or $\tilde{t}[i][j]$) $\notin \mathcal{V}$}{
		        
		            $\mathcal{V}$.add($\tilde{t}[j][i]$)\;
		            
		        }\ElseIf{$count[\tilde{t}[j][i]-1] \geq 1$}{
		            $count[\tilde{t}[j][i]-1] \leftarrow count[\tilde{t}[j][i]-1] - 1$\;
		            
		            \If{$count[\tilde{t}[j][i]-1] = 0$}{
		                $\mathcal{V}$.remove($\tilde{t}[j][i]]$)\;
		            }
		        }
		}
		    
		$L$.add($i$)\;
		    
		\For{$k \in i$.outNeighbors()}{
		    \textsc{updateShortestPaths}($L$, $dist$, $k$, $T_x$, $d$, $\tilde{t}$, $\mathcal{V}$, $count$)\;
		}
		$L$.remove($i$)\;
	
		\caption{\textsc{updateShortestPaths($L$, $dist$, $i$, $T_x$, $d$, $\tilde{t}$, $\mathcal{V}$, $count$)}}
		\label{alg:distances1}
\end{algorithm}

%----

%--
\begin{lemma}\label{lema:lista_L}
At the beginning of every call to Algorithm \ref{alg:distances1}, the list $L$ contains the  predecessors of $i$ in the path from $x$ to $i$ in the tree $T_x$. 
\end{lemma}

\begin{proof}  We use induction on the size of $L$. Let $P_{xi}$ be the set of predecessors of $i$ in the path from $x$ to $i$ in the tree $T_x$. If
 $L=\{\}$ (i.e., \textsc{updateShortestPaths} is called by Algorithm \ref{alg:distances1} and has not yet recursevily called itself), then $i=x$. As $P_{xx} = \{\}$, then $L=P_{xx}$, and the base case follows.
 
 In the inductive step, we show that the result follows for the recursive calls of \textsc{updateShortestPaths} (line 15). Let $L'$ and $k$ be, respectively, the list and the vertex used as arguments on the recursive calls. As $k$ is a successor of $i$, then $P_{xk} = P_{xi} \cup \{i\}$. But, from line 13, $L' = L \cup \{i\}$. By the induction hypothesis $L = P_{xi}$ and therefore $L' = P_{xi} \cup \{i\} = P_{xk}$.
\end{proof}

\begin{algorithm}[!htbp]
		\SetAlgoNoLine
		\SetAlgoNoEnd
		\DontPrintSemicolon	
		\KwData{Weighted graph $G = (V,E)$ with $n = |V|$ and $m = |E|$, accuracy parameter $0 < \epsilon < 1$, confidence parameter $0 < \delta < 1$, sample scheduling $(S_i)_{i \geq 1}$}
		
		%\KwResult{Approximation $\tilde{w}(p_{uv})$, for each $(u,v) \in V^2$.}
		\KwResult{Probabilistic Shortest Path Distance $d[u][v]$, for each $(u,v) \in V^2$ with $d[u][v] > 0$.}
		
		%$M \leftarrow n \times n$ matrix\; %matriz das distancias exatas
	    
	    \iffalse\For{$(u,v) \in V^2$}{
	        
	        $\tilde{c}[u][v] \leftarrow \tilde{t}[u][v] \leftarrow 0$\;
	        
	        \If{$u = v$}{ 
	            $d[u][v] \leftarrow 0$\;
	        }\Else{
	            $d[u][v] \leftarrow \infty$\;
	        }
	    }\fi
		
		$count[i] \leftarrow 0, \forall i \in {1,\ldots,n}$\;
		
		$|S_0| \leftarrow 0$\;
		
		$\mathcal{V} \leftarrow \emptyset$\;

		$i \leftarrow 0, j \leftarrow 1$\;
		
		\Do{$\eta > \epsilon$}{
		    
		    $i \leftarrow i + 1$\;
		    
		    \For{$l \leftarrow 1 $ \KwTo $|S_i| - |S_{i-1}|$}{
		        sample $x \in V$ with probability $1/n$\;
		    
		        $T_x,\, dist \leftarrow \textsc{singleSourceShortestPaths}(x)$\; %árvore Dijkstra 'T_x' e vetor de distancias 'dist'
		    
		        $L \leftarrow$ empty list\;
		    
		        \textsc{updateShortestPath}($L$, $dist$, $x$, $T_x$, $d$, $\tilde{t}$, $\mathcal{V}$, $count$)\;
		    
		    }
		    
		    \iffalse\If{$spc = 0$ and $|S_i| \geq \lceil \frac{c}{\epsilon} \lfloor \lg 2 + \lg \textrm{Diam}_V(G) + \lg n \rfloor \ln \frac{1}{\epsilon} + 1 + \ln \frac{1}{\delta} \rceil$}{
		        $\eta \leftarrow \epsilon$
		    }\fi%\Else{
		    
		    $w_s \leftarrow \min_{s \in \mathbb{R}^+} \frac{1}{s} \ln \sum_{t \in \mathcal{V}} \exp{\frac{s^2 t}{2|S_i|^2}}$\;
		    
		    $\delta_i \leftarrow \delta/2^i$\;
		        
		    $\eta \leftarrow 2w_s + \frac{\ln \frac{3}{\delta_i} + \sqrt{(\ln \frac{3}{\delta_i} + 4|S_i| w_s) \ln \frac{3}{\delta_i}}}{|S_i|} + \sqrt{\frac{\ln \frac{3}{\delta_i}}{2|S_i|}}$}
		    
		%}
		
		$\tilde{c}[u][v] \leftarrow \tilde{t}[u][v]/|S_i|, \text{ for each } (u,v) \in V^2$ such that $\tilde{t}[u][v] > 0$\;
		\If{$spc = 1$}{
		    \textbf{return }{$d[u][v] \text{ and } \tilde{c}[u][v], \,\, \text{ for each} (u,v) \in V^2$ such that $\tilde{t}[u][v] > 0$}
		}\Else{
		    \textbf{return }{$d[u][v], \,\, \text{ for each } (u,v) \in V^2$ such that $\tilde{d}[u][v] > 0$}
		}
		\caption{\textsc{ProbabilisticAllPairsShortestPaths($G$,$\epsilon$,$\delta$,$spc$)}}
		\label{alg:distances2}
\end{algorithm}

\begin{theorem}\label{theo:distances}
The value $d[i][j]$ (resp. $d[j][i]$) set by Algorithm \ref{alg:distances1} is equal to the distance between vertices $i$ and $j$ (resp. $j$ and $i$).    
\end{theorem}

\begin{proof}
Let $d_{st}$ be the distance between vertices $s$ and $t$. 
The algorithm sets $d[i][j] = d_{xi} - d_{xj}$. Thus, it suffices to prove that $d_{xi} - d_{xj} = d_{ij}$. Suppose by contradiction that $d_{xi} - d_{xj} \neq d_{ij}$, and split this supposition in two cases: (i) $d_{xi} - d_{xj} > d_{ij}$ and (ii) $d_{xi} - d_{xj} < d_{ij}$. In case (i), the path from $x$ to $i$ can be traversed in two parts, from $x$ to $j$ and from $j$ to $i$. Hence, the total distance from $x$ to $i$ is $d_{xj} + d_{ij} < d_{xj} + d_{xi} - d_{xj} = d_{xi}$, contradicting the fact that $d_{xi}$ is minimum. In case (ii), by Algorithm \ref{alg:distances1} we have that $j \in L$, and therefore (Lemma \ref{lema:lista_L}) $j$ is a predecessor of $i$. So, the shortest path between $x$ and $i$ necessarily passes through $j$. Let $p_{xi}$ be the shortest path between $x$ and $i$. Thus, $i$ can reach $j$ through the path $p_{xi}$ and in this case the distance is $d_{xi} - d_{xj}$. But, by the hypothesis of case (ii), $d_{xi} - d_{xj} < d_{ij}$, contradicting the fact that  $d_{ij}$ is minimum. 
\end{proof}

\begin{corollary}\label{coro:paths}
    All shortest paths in $S(T_x)$ can be computed by a modification of Algorithm \ref{alg:distances1}.
\end{corollary}

\begin{proof}
    We can store $z$ alongside with $d[j][i]$ in line 2, where $z$ is the immediate predecessor of $i$ in the shortest path $p_{xi}$, for each $j \in L$.
\end{proof}

\begin{theorem}\label{theo:alg2}

    Consider a sample $S_r = \{T_1,\ldots,T_r\}$ of size $r$ and let $\eta_i$ be the value obtained in line 14 on the $i$-th iteration. Then $\eta_r = 2w_s + \frac{\ln \frac{3}{\delta_r} + \sqrt{(\ln \frac{3}{\delta_r} + 4|S_r|w_s) \ln \frac{3}{\delta_r}}}{|S_r|} + \sqrt{\frac{\ln \frac{3}{\delta_r}}{2|S_r|}}$, where $\delta_r = \delta/2^r$, is the value where $r$ is the minimal $i \geq 1$ such that $\eta_r \leq \epsilon$ % = \lceil \frac{c}{\epsilon^2}\left(2 \lfloor \ln n \rfloor + 1 - \ln \delta \right) \rceil$ 
    for the input graph $G = (V,E)$ and for fixed constants $0 < \epsilon,\delta < 1$. Algorithm \ref{alg:distances2} returns with probability at least $1-\delta$ the exact distance $d(u,v)$, for each $(u,v) \in V^2$ such that $d(u,v) > 0$, and corresponding shortest path between the vertices $u$ and $v$ whenever $p_{uv}$ has centrality at least $\epsilon$. Additionally, the value of $\tilde{c}(u,v)$ is within $\epsilon$ error to the value of $c(u,v)$ with probability $1-\delta$.
\end{theorem}

\begin{proof}
    Let $i \geq 1$ be an iteration of the loop in 5--15 and let $E_i$ be the event where $\sup_{(u,v) \in V^2}|\tilde{c}(u,v) - c(u,v)| > \eta_i$ in this iteration. We need the event $E_i$ occurring with probability at most $\delta$ for some iteration $i$. That is, we need 
    
    $$\Pr(\exists i \geq 1 \text{ s.t. } E_i \text{ occurs}) \leq \sum_{i=1}^\infty \Pr(E_i) \leq \delta$$

    where the inequality comes from union bound. Setting $\Pr(E_i) = \delta/2^i$, we have 
    
    $$\sum_{i=1}^\infty \Pr(E_i) = \delta \sum_{i=1}^\infty \frac{1}{2^i} = \delta.$$

    Let $S_r = \{T_1,\ldots,T_r\}$ be the final sample obtained after the iteration $r$ in the loop 5--15 where the stopping condition is satisfied, i.e., $\eta_r \leq \epsilon$. For each iteration $i$ in 5--15, where $1 \leq i \leq r$, consider that for each $x \in V$, there is one Dijkstra tree $T_x \in \mathcal{T}$, and hence, $|\mathcal{T}| = n$. A vertex $x \in V$ is sampled with probability $1/n$; therefore, $T_x$ is sampled with probability $1/n$ (line 8).
    %(lines 9 and 10). 
    The tree $T_x$ is traversed by Algorithm \ref{alg:distances1}, and the distances of every shortest path $p_{xy} \in S(T_x)$ are correctly and exactly computed, as shown in Lemma \ref{lema:lista_L}, Theorem \ref{theo:distances} and Corollary \ref{coro:paths}. 
    
    Let $p_{uv} \in S(T_x)$ be a shortest path from $u$ to $v$ in the sampled tree. At this point, $d(u,v)$ has the exact distance and the shortest path from $u$ to $v$ correctly computed. We will show by contraposition that if $c(u,v) \geq \epsilon$, then $|\tau_{uv} \cap S| \geq 1$, where $\tau_{uv} = \{T_x \in \mathcal{T} : p_{uv} \in S(T_x)\}$ and $c(u,v) = \mathbb{E}[f_{uv}(T_x)] = \Pr_\pi(\tau_{uv})$. 
    
    If $|\tau_{uv} \cap S| < 1$, then there is no tree $T'_x \in S$ that contains a shortest path from $u$ to $v$, and hence $\tilde{c}(u,v) = 0$. Then the value of $c(u,v)$ must be at most $\epsilon$ so that $|\tilde{c}(u,v) - c(u,v)| \leq \epsilon$ holds. Therefore, if $c(u,v) \geq \epsilon$, then $|\tau_{uv} \cap S| \geq 1$ and $S$ with size at most $\lceil \frac{c}{\epsilon}\lfloor\lg \textrm{Diam}_V(G) + \lg n + 2 \rfloor \ln \frac{1}{\epsilon} + \ln \frac{1}{\delta} \rceil$ is an $\epsilon$-net with probability at least $1-\delta$. The probability that $d(u,v)$ (as well as $d(v,u)$) is exactly computed is $ \geq 1 - \delta$ (Theorem \ref{teo:esample} (ii)).

    Now consider the computation of the estimation $\tilde{c}(u,v)$, for a pair $(u,v) \in V^2$. Let $p_{uv} \in S(G)$ be a shortest path from $u$ to $v$ and let $S' \subseteq S_r$ be the set of trees such that $p_{uv} \in S(T'_x)$, where $T'_x \in S'$. If the tree sampled in line 8
    %and 10 
    of Algorithm \ref{alg:distances2} is in $S'$, then the value $\tilde{t}(u,v)$ has its value increased by 1 in line 5 of Algorithm \ref{alg:distances1}, so at the end or $r$-th iteration in loop 11--23, $\tilde{c}(u,v) = \frac{\tilde{t}(u,v)}{|S_r|} = \frac{1}{|S_r|}\sum_{T'_x \in S'} 1 = \frac{1}{|S_r|} \sum_{T_i \in S} \mathds{1}_{\tau_{uv}}(T_i) = \frac{1}{|S_r|} \sum_{T_i \in S} f_{uv}(T_i)$.
    
    Since $\eta_r \leq \epsilon$, $L_S(f_{uv}) = \tilde{c}(u,v)$ and $L_U(f_{uv}) = c(u,v)$ (Theorem \ref{theo:expec}) for all $(u,v) \in V^2$ and $f_{uv} \in \mathcal{F}$, then $\Pr(|\tilde{c}(u,v) - c(u,v)| \leq \epsilon) \geq 1-\delta$ (Theorem \ref{theo:rademacherbound}).    
\end{proof}

%The convex optimization in line 12 of Algorithm \ref{alg:distances2} is very efficient in practice \cite{RiondatoUpfal}, however its theoretical complexity analysis needs further investigation. Therefore, for the sake of simplicity, we will give a time complexity analysis for a variant of Algorithm \ref{alg:distances2} using a fixed size sample. 

\begin{theorem}\label{theo:time}
    Given an undirected weighted graph $G=(V,E)$ with $n = |V|$ and a sample of size at most $r = \lceil \frac{c}{\epsilon}\lfloor\lg \textrm{Diam}_V(G) + \lg n + 2\rfloor \ln \frac{1}{\epsilon} + \ln \frac{1}{\delta} \rceil$, %(or $\lceil \frac{c}{\epsilon^2} \lfloor \lg 2 + \lg Diam(G) + \lg n \rfloor + 1 + \ln \frac{1}{\delta} \rceil$),
    Algorithm 2 has expected running time $\mathcal{O}(\lg n \max(m + n \log n, n \cdot \textrm{Diam}_V(G))$ for the computation of table $d$.
\end{theorem}

\begin{proof}
    We sample the vertex $x \in V$ in line 8 in linear time. 
    Line 11 takes time $\mathcal{O}(n \cdot \textrm{Diam}_V(G))$ because Algorithm \ref{alg:distances1} makes $O(n)$ recursive calls -- one for each vertex of the canonical tree -- and the loop execution in line 1 in Algorithm \ref{alg:distances1} takes time $\mathcal{O}(\textrm{Diam}_V(G))$ (since in the worst case, the Dijkstra tree $T_x$ cannot be deeper than the diameter of the graph). Line 12 is executed by an algorithm that is linear in the size of the sample \cite{nlopt}.
    The loop in lines 5--15 runs at most $r$ times and the Dijkstra algorithm which is executed in line 9 has running time $\mathcal{O}(m + n \log n)$. The operations of insertion, deletion and search in tables $\tilde{t}$ and $d$ take time $\mathcal{O}(1)$ on average. So, the total expected running time of Algorithm \ref{alg:distances2} is $\mathcal{O}(r\max(m + n \log n, n \cdot \textrm{Diam}_V(G))) = \mathcal{O}(\lg n \max(m + n \log n, n \cdot \textrm{Diam}_V(G)))$.
\end{proof}

\begin{corollary}\label{corol:time}
    Given an undirected weighted graph $G=(V,E)$ with $n = |V|$ and a sample of size $r = \lceil \frac{c}{\epsilon^2}\lfloor\lg \textrm{Diam}_V(G) + \lg n + 2\rfloor \ln \frac{1}{\delta} \rceil$, %(or $\lceil \frac{c}{\epsilon^2} \lfloor \lg 2 + \lg Diam_V(G) + \lg n \rfloor + 1 + \ln \frac{1}{\delta} \rceil$),
    Algorithm 2 has running time $\mathcal{O}(\lg n \max(m + n \log n, n\cdot\textrm{Diam}_V(G)))$ for the computation of table $\tilde{c}$.
\end{corollary}
%------------------------------------
%------------------------------------

\section{Concluding remarks}

In this paper we presented a $\mathcal{O}(\lg n \max(m + n \log n, n \cdot \textrm{Diam}_V(G)))$ expected running time algorithm that 
%for every pair of vertices $(u,v)$ output an estimation for a measure $c(u,v)$, called the centrality of a shortest path between $u$ and $v$. The output is within $\epsilon$ of the exact value with probability $1-\delta$, for fixed constants $0 < \epsilon, \delta < 1$. We show that this algorithm can be adapted to 
outputs a shortest path between every pair of vertices $(u,v)$ with probability at least $1-\delta$ whenever the shortest path centrality of $(u,v)$ is at least $\epsilon$, for fixed $0 < \epsilon,\delta < 1$. We note that this is particularly interesting in sparse graphs with logarithmic diameter, such as power law graphs, which are common in practical applications, since for those graphs the complexity drops to $\mathcal{O}(n \log^2 n)$. So, in an application where one might be interested only in computing ``central'' shortest paths the algorithm is rather efficient.

%, so we include some preliminary experiments in the appendix in order to clarify some aspects of the practicality of this approach. In particular we address two issues: 

%\begin{itemize}
%    \item[(1)] Since it may not be clear how, in practice, the threshold $\epsilon$ for the centrality impacts on the fraction of shortest paths computed we performed experiments on real and synthetic graphs in order to show how different values of $\epsilon$ impact on the amount of shortest paths computed.
%    \item[(2)] Since the centrality of a shortest path depends on the ordering of the vertices, one may question the validity of the definition of centrality. Even though in a sampling algorithm it is somewhat expected that the ordering does not have a big impact, we perform experiments in order to confirm this fact.
%\end{itemize}

%Finally, we note that both (1) and (2) above are very interesting problems that can be explored analytically if we assume that the input graph is sampled from a given distribution (or are known to belong to some specific graph class).  

Finally, an open question that we are particularly interested is the connection between $\epsilon$ and $n$ for specific input distributions. For the general case, trivially setting $\epsilon = 2/n$, by Theorem 1, we have a guarantee that every shortest path in $G$ is computed with probability $1-\delta$, but that would increase the algorithm complexity to $\tilde{O}(n^3)$. We wonder if this fact may be related to the assumption that the APSP may not admit a strictly subcubic algorithm. However, if we assume that the graph is sampled from a given probability distribution, a strictly subcubic randomized algorithm for the (original) APSP may be achievable.

%, suggesting that our approach is practical in real-world graphs. %Similar results for directed graphs is left as an open problem.

%\bibliographystyle{splncs04}
\bibliography{lipics-v2019-sample-article}

\end{document}